\documentclass[submission,copyright,creativecommons]{eptcs}
\providecommand{\event}{ICLP 2025} 

\usepackage{iftex}

\ifpdf
  \usepackage{underscore}         
  \usepackage[T1]{fontenc}        
\else
  \usepackage{breakurl}           
\fi

\usepackage{amsfonts}
\usepackage{amssymb}
 \usepackage{amsmath}
 \usepackage{amsthm}
\usepackage{latexsym}
\usepackage{stmaryrd}
\usepackage{microtype}
\usepackage{comment}
\usepackage{preamble}

\usepackage{scalerel,stackengine}

\newtheorem{lemma}{Lemma}[section]
\newtheorem{proposition}{Proposition}[section]
\newtheorem{definition}{Definition}[section]
\newtheorem{theorem}{Theorem}[section]

\newtheorem{example}{Example}[section]
\newtheorem{remark}{Remark}

\newcommand{\mtrue}{\mathit{true}}
\newcommand{\mfalse}{\mathit{false}}
\newcommand{\mundef}{\mathit{undef}}

\newcommand{\lub}{\bigsqcup}

\newcommand{\glb}{\bigsqcap}

\newcommand\leqp{\leq_p}
\newcommand\lrule{\leftarrow}

\usepackage{xspace}

\newcommand\ie{i.e.,\xspace}

\usepackage[final]{listings}
\usepackage{caption}

\usepackage{amsmath}
\DeclareMathOperator\lfp{lfp}

\newcommand{\timeline}{\mathbb{T}}
\font\Bigmath=cmsy10 scaled \magstep3
\def\diamondplus{\mathrel{%
		\ooalign{$+$\cr\hss\lower.3ex\hbox{\Bigmath\char5}\hss}}}
\def\diamondminus{\mathrel{%
		\ooalign{$-$\cr\hss\lower.3ex\hbox{\Bigmath\char5}\hss}}}
\newcommand{\interval}{\delta}
\newcommand{\negation}{\mathsf{not}\;}
\newcommand{\ground}[1]{\mathsf{ground}(#1)}
\newcommand{\prog}{\Pi}
\newcommand{\semtwo}[3]{\llbracket#1\rrbracket(#2,#3)}
\newcommand{\semcal}[3]{\llbracket#1\rrbracket(\mathcal{#2},#3)}
\newcommand{\semthree}[4]{\llbracket#1\rrbracket(#2,#3,#4)}
\newcommand{\db}{\mathcal{D}}
\newcommand{\immed}[2]{T_{#1,#2}}
\newcommand{\interpret}{\mathbb{I}}
\newcommand{\intercons}{\interpret^{\mathsf{c}}}
\newcommand{\appr}[2]{\Psi_{#1,#2}}
\newcommand\equalhat{\mathrel{\stackon[1.5pt]{=}{\stretchto{%
				\scalerel*[\widthof{=}]{\wedge}{\rule{1ex}{3ex}}}{0.5ex}}}}
\newcommand{\headatom}{\mathcal{M}}
\newcommand{\leqtr}{\leq_{\tau}}
\newcommand{\geqtr}{\geq_{\tau}}

\title{Fixpoint Semantics for DatalogMTL with Negation\thanks{This study was funded
	by Fonds Wetenschappelijk Onderzoek -- Vlaanderen (project G0B2221N)}}
\author{Samuele Pollaci
\institute{Vrije Universiteit Brussel\\ Brussels, Belgium}
\institute{Katholieke Universiteit Leuven\\
Leuven, Belgium}
\email{Samuele.Pollaci@vub.be}
}

\newcommand{\titlerunning}{Fixpoint Semantics for DatalogMTL with Negation}
\newcommand{\authorrunning}{S. Pollaci}

\hypersetup{
  bookmarksnumbered,
  pdftitle    = {\titlerunning},
  pdfauthor   = {\authorrunning},
  pdfsubject  = {EPTCS},               
}

\begin{document}
\maketitle

\begin{abstract}
DatalogMTL$^\neg$ is an extension of Datalog with metric temporal operators enriched with unstratifiable negation. In this paper, we define  the stable, well-founded, Kripke-Kleene, and  supported model semantics for DatalogMTL$^\neg$ in a very simple and straightforward way, by using the solid mathematical formalism of  Approximation Fixpoint Theory (AFT). Moreover, we prove that the stable model semantics obtained via AFT coincides with  the one defined in 
previous work, through the employment of pairs of interpretations stemming from the logic of here-and-there. 
\end{abstract}

\section{Introduction}

Datalog \cite{AbiteboulHV95} is a declarative query language that extends relational algebra with recursion. The last few decades have witnessed the emergence of several variants of Datalog and their implementations \cite{KetsmanK22}, aiming to solve data management tasks in various applications \cite{HuangGL11}. One of such variants is DatalogMTL \cite{BrandtKRXZ18}, which extends the Datalog language with operators from metric temporal logic \cite{Koymans90}, and finds applications in ontology-based query answering \cite{BrandtKRXZ18,KalayciXRKC18,KikotRWZ18,Koopmann19}, stream reasoning \cite{WalegaKG19,WalegaKWG23}, and reasoning for the financial sector \cite{NisslS22,MoriPBG22,ColomboBCL23}, amongst others.

For instance, we can use the following DatalogMTL rule to express that an adult has to wait at least 6 hours after the assumption of a paracetamol tablet, before taking another one:

\begin{equation*}
	\mathit{NoMoreParacetamol}(x) \lrule \mathit{Adult}(x) \wedge \diamondminus_{[0,6]}\mathit{TakesParacetamol}(x),
\end{equation*}
where $\diamondminus_{[0,6]}\mathit{TakesParacetamol}(x)$ is true at a timepoint $t$, if $\mathit{TakesParacetamol}(x)$ was true at some point in the last 6 hours, i.e. at some time $t'$ in the interval $[t-6,t]$. Other than $\diamondminus_{[0,6]}$, there are a few more metric temporal operators that are allowed in the expressions of DatalogMTL. 
For instance, if $\varphi$ is a formula and $[t_1,t_2]$ is an interval on the rational or integer timeline, then $\boxminus_{[t_1,t_2]}\varphi$ holds at time $t$ if $\varphi$ holds at all the timepoints of the interval $[t-t_2,t-t_1]$. 
Analogously, there are also the future-oriented counterparts of $\diamondminus$ and $\boxminus$, namely $\diamondplus$ and $\boxplus$, respectively: $\diamondplus_{[t_1,t_2]}\varphi$ holds at time $t$ if $\varphi$ holds at some timepoint of the interval $[t+t_1,t+t_2]$; and similarly for $\boxplus_{[t_1,t_2]}\varphi$. 
A DatalogMTL dataset is then a set of facts, each paired with a time interval. For example, if 
$\mathit{TakesParacetamol}(\mathit{John})@[8,8]$ 
is in a dataset, then it holds that $\mathit{John}$ took a paracetamol tablet at $8$a.m.

In 2021,  DatalogMTL was extended with stratified negation as failure \cite{CucalaWGK21}, allowing negated atoms in the body of rules, but no recursion. Recently, the range of capabilities of DatalogMTL was further expanded by equipping DatalogMTL with unstratifiable negation under the stable model semantics \cite{WalegaCGK24}, while extending both DatalogMTL with stratified negation and Datalog with stable model negation. In such paper,
to obtain the stable models of a dataset and a program, the authors employed interpretations similar to the ones used for the logic of \emph{here-and-there} (HT)  \cite{Heyting30,Pearce96}, namely pairs $(I,J)$ of (two-valued) interpretations such that the atoms evaluated as \emph{true} by $I$ are $\mtrue$ under $J$ as well. Although used differently, this same framework with pairs of interpretations sits at the core of numerous applications of Approximation Fixpoint Theory (AFT) \cite{DeneckerMT00,DeneckerMT04}, prompting the creation of this paper.

AFT was first developed to study the semantics of non-monotonic logics, such as logic programming, autoepistemic logic, and default logic. It has been successfully applied in several other areas, like abstract argumentation \cite{Strass13}, active integrity constraints \cite{BogaertsC18}, stream reasoning \cite{Antic20}, and constraint languages for the semantic web \cite{BogaertsJ21}.
In the context of (consistent) AFT \cite{DeneckerMT04}, a  pair $(I,J)$ of two-valued interpretations with $I\subseteq J$ is viewed as an  \emph{approximation} of some other two-valued interpretation: an atom interpreted as $\mtrue$ (resp.\ $\mfalse$) by both $I$ and $J$ has a precise, exact value, namely $\mtrue$ (resp. $\mfalse$); while an atom assigned to $\mfalse$ by $I$ but to $\mtrue$ by $J$ can be either $\mtrue$ or $\mfalse$, i.e.\ it is \emph{undefined}. In other words, a pair of two-valued interpretations is equivalent to a \emph{three-valued interpretation} sending each atom to either $\mtrue$, $\mfalse$, or $\mundef$. Depending on the application, a suitable semantic operator, called an \emph{approximator}, is defined on the set of three-valued interpretations, and then the machinery of AFT provides a family of well-known semantics and desired properties straightforwardly and clearly. It turns out that consistent AFT is suitable to define the stable model semantics for DatalogMTL$^\neg$, showcasing once again the power and versatility of the AFT formalism. 

In this paper, we consider the set of pairs of interpretations already considered in \cite{WalegaCGK24}, define an operator on it, and prove that such operator is a consistent approximator. Then, AFT allows us to obtain not only the \emph{stable} model semantics for DatalogMTL$^\neg$, but also the \emph{well-founded}, the \emph{Kripke-Kleene}, and the \emph{supported} model semantics as well. Finally, we prove that the stable models obtained via the usage of AFT are the same as the stable models defined in \cite{WalegaCGK24}.
The advantages of employing the framework of AFT may be summarized as follows:

\begin{enumerate}
	\item \emph{Effortless.} The only thing required, is defining the right \emph{approximator}; then the machinery of AFT lays out the stable, the well-founded, the Kripke-Kleene, and the supported model semantics for free.
	For most applications, the choice of an approximator is very natural;  for DatalogMTL$^\neg$, it is no exception. The approximator, in fact, closely resembles Fitting's three-valued immediate consequence operator \cite{Fitting02}, with some slight adjustments to account for time points and datasets.
	\item \emph{Solid mathematical background.} AFT provides a broad range of properties and general results that hold for all its applications; hence, there is no need to re-prove several statements when switching settings. For instance, the minimality and uniqueness result (Theorem 3.4 in  \cite{WalegaCGK24}) leading to the definition of stable models for DatalogMTL$^\neg$ in \cite{WalegaCGK24}, is already guaranteed by the AFT framework. 
	\item \emph{Reliable.} The strong mathematical foundations of AFT make the obtained results more reliable. Even when semantics have already been defined, as in the case of the stable model semantics for DatalogMTL$^\neg$, AFT constitutes a valuable tool to screen the body of previous work for errors.  For example, the employment of AFT in the context of Abstract Dialectical Frameworks allowed the detection of some bugs in the original semantics \cite{Bogaerts19}.
\end{enumerate}

The outline of the paper is as follows. First, we provide some preliminary material on both AFT (Section \ref{sec:prel:AFT}) and the syntax and semantics of DatalogMTL$^\neg$ (Section \ref{sec:prel:datalog}). In Section \ref{sec:threeval}, we define the base notions we need to apply AFT: the three-valued semantics, and the immediate consequence operators (two- and three-valued). Moreover, we prove there exists a correspondence between the pre-fixpoints of these  operators, and the two-valued and three-valued models. In Section \ref{sec:AFTdatalog}, we are finally able to apply AFT; we just need to prove that the three-valued immediate consequence operator defined in the previous section is indeed a consistent approximator. Then, we are able to define the Kripke-Kleene, stable, well-founded, and supported models via AFT. Finally, in Section \ref{sec:comparison}, we compare our notions of three-valued models and two-valued stable models, with the concepts from \cite{WalegaCGK24}, and we show that the stable models obtained via AFT are the same as the ones from \cite{WalegaCGK24}.

\section{Preliminaries}

In this section, we first present the concepts that lays the foundations of AFT (Section \ref{sec:prel:AFT}). Then, we introduce the basic definitions and notation relative to timelines, the syntax of DatalogMTL$^\neg$, and the semantics of rule bodies (Section \ref{sec:prel:datalog}), as they were presented in \cite{WalegaCGK24}.

\subsection{AFT}\label{sec:prel:AFT}

AFT generalizes Tarki's theory to non-monotonic operators \cite{Tarski55}, with the initial goal of studying the semantics of non-monotonic logics. As such, AFT heavily relies on the following notions from order theory.

A \emph{partially ordered set} (poset) $P$ is a set equipped with a partial order $\leq_P$, \ie a reflexive, antisymmetric, transitive relation. We often denote a poset by $\langle P, \leq_P \rangle$, where $P$ is the underlying set, and $\leq_P$ the partial order. 
Given a subset $S\subseteq P$, a lower bound $l$ of $S$ is the \emph{greatest lower bound of $S$}, denoted by $\glb S$, if it is greater than any other lower bound of $S$. Analogously, an upper bound $u$ of $S$ is the \emph{least upper bound of $S$}, denoted by $\lub S$, 
if it is lower than any other upper bound of $S$. 
A \emph{chain complete poset} (cpo)  is a poset $C$ such that for every chain $S\subseteq C$, \ie a totally ordered subset, $\lub S$ exists. A \emph{complete lattice} is a poset $L$ such that for every subset $S\subseteq L$, both $\lub S$ and $\glb S$ exist. 
A function $f\colon P_1\to P_2$ between posets is \emph{monotone} if for all $x,y\in P_1$ such that $x\leq_{P_1}y$, it holds that $f(x)\leq_{P_2}f(y)$. 
We refer to functions $O\colon C\to C$ with domain equal to the codomain as \emph{operators}. An element $x\in C$ is a \emph{fixpoint} (resp.\ \emph{pre-fixpoint})
of $O$ if $O(x)=x$ (resp.\ $O(x)\leq_C x$). 
By Tarski's least fixpoint theorem, every monotone operator $O$ on a cpo has a least fixpoint, denoted $\lfp(O)$. To use a similar principle for operators stemming from non-monotonic logics, consistent AFT \cite{DeneckerMT04} considers, for each complete lattice $L$, its associated \emph{approximation space} $\langle L^c, \leqp\rangle$, where $L^c=\{(x,y)\mid x,y\in L, x\leq_L y\}\subseteq L^2$, and $\leqp$ is the \emph{precision order} on the Cartesian product $L^2$, \ie  $(x_1,y_1)\leqp (x_2,y_2)$ iff $x_1\leq_L x_2$ and $y_2\leq_L y_1$. 
 $L^c$ can be viewed as an approximation of $L$: an element $(x,y)\in L^c$ ``approximates'' all the values $z\in L$ such that $x\leq_L z\leq_L y$. Pairs of the form $(x,x)\in L^2$ are called \emph{exact}, since they approximate only one element of $L$. 

An operator $A\colon L^c\to L^c$ is a  \emph{consistent approximator  of an operator} $O\colon L\to L$ if it is monotone and it behaves as $O$ on exact pairs, i.e.\  $A(x,x)=(O(x),O(x))$ for all $x\in L$. We denote by $A_1$ and $A_2$ the first and second component of $A$, respectively, i.e.\ $A(x,y)=(A_1(x,y),A_2(x,y))$ for all $(x,y)\in L^c$. Since $A$ is by definition monotone, by Tarski's theorem $A$ has a least fixpoint, which is called the \emph{Kripke-Kleene} fixpoint. Moreover, given an approximator $A$, it is easy to define three additional kinds of fixpoints of interests: \emph{supported}, \emph{stable}, and \emph{well-founded} fixpoints. This will be done in detail at the end of Section \ref{sec:AFTdatalog}, in Definition \ref{def:models}.
	If $A$ is Fitting's three-valued immediate consequence operator \cite{Fitting02}
	, then the aformentioned four types of fixpoint correspond to the homonymous semantics of logic programming \cite{DeneckerMT00,DeneckerBV12}. 

\subsection{DatalogMTL$^\neg$}\label{sec:prel:datalog}

A \emph{timeline} $\timeline$ is either the ring of integers $\langle \mathbb{Z}, +\rangle$, or the ring of rationals $\langle \mathbb{Q}, +\rangle$, with the order given by the addition. A \emph{($\timeline$-)interval} $\interval$ is a non-empty subset of $\timeline$ such that the following conditions are satisfied:
\begin{enumerate}
	\item $\interval$ is \emph{convex}, \ie for all $t_1, t_2\in \interval$, for all $t\in \timeline$ such that $t_1<t<t_2$, it holds that $t\in \interval$;
	\item $\glb \interval, \lub \interval \in \timeline\cup \{+\infty, -\infty\}$.
\end{enumerate}
We call an interval $\interval$ \emph{non-negative} if for all $t\in \interval$, it holds that $t\geq 0$. We represent an interval $\interval$ with the standard notation $\langle \glb \interval, \lub \interval\rangle$, where the left bracket $\langle$ can either be $[$, if $\glb\interval\in\interval$, or $($ otherwise; and the right bracket $\rangle$ can either be $]$, if $\lub\interval\in\interval$, or $)$ otherwise. For any $t\in \timeline$, we denote the interval $[t,t]$ by simply $t$.

We now consider a function-free first-order vocabulary and a timeline $\timeline$. For any predicate $P$, and tuple $\mathbf{s}$ of constants and variables of the same arity of $P$, we call $P(\mathbf{s})$ a \emph{relational atom}. A \emph{metric atom} is an expression given by 
\begin{equation*}
	M ::= \top \mid \bot \mid P(\mathbf{s}) \mid \diamondminus_\interval M \mid \diamondplus_\interval M \mid \boxminus_\interval M \mid \boxplus_\interval M \mid M \mathcal{S}_\interval M \mid M \mathcal{U}_\interval M,
\end{equation*}
where $P(\mathbf{s}) $ ranges over relational atoms, and $\interval$ ranges over non-negative intervals. A metric atom is \emph{ground} if it does not contain any variable. A \emph{head atom} is a metric atom  specified by the grammar:

\begin{equation} \label{eq:atomhead}
	M ::= \top \mid P(\mathbf{s}) \mid  \boxminus_\interval M \mid \boxplus_\interval M,
\end{equation}
where $P(\mathbf{s}) $ ranges over relational atom, and $\interval$ ranges over non-negative intervals. We denote by $\headatom$ the set of ground head atoms. A \emph{metric fact} is an expression $M@\interval$, where $M$ is a ground metric atom, and $\interval$ is an interval. A metric fact $M@\interval$ is also a  \emph{relational fact} if $M$ is a ground relational atom. A \emph{dataset} is a finite set of relational facts.
	A \emph{(two-valued) interpretation} $I$ is a function which assigns to each $t\in \timeline$ a set $I(t)$ of ground relational atoms.
	A two-valued interpretation $I$ is a \emph{(two-valued) model of a dataset} $\db$ if for all relational facts $P(\mathbf{s})@\interval\in\db$, and for all $t\in\interval$, it holds that $P(\mathbf{s})\in I(t)$. 
 Notice that, interpretations can be ordered with the subset relation, \ie $I\subseteq J$ if and only if for all $t\in\timeline$, $I(t)\subseteq J(t)$. We denote by $\interpret$ the set of all two-valued interpretations.
	A \emph{rule} $r$ is an expression of the form 
	\begin{equation}\label{eq:rule}
		M \lrule M_1\wedge\ldots \wedge M_k\wedge \negation M_{k+1}\wedge \ldots \wedge \negation M_m, 
	\end{equation}
where $m\geq k\geq 0$, each  $M_i$ is a metric atom, and $M$ is a head atom. 
 We call the atoms $M_1, \ldots, M_k$ \emph{positive}, and $M_{k+1}, \ldots, M_m$ \emph{negated}. The conjunction of positive and negated atoms of a rule is called its \emph{body}, whereas the consequent $M$ is its \emph{head}. For brevity reasons, we will often denote a rule as $M\lrule B$, where $B$ represents the conjunction in the body. 
 Notice that rules of the form $\bot\lrule B$ are allowed in \cite{WalegaCGK24}, but not in this paper. We will instead use rules of the form $P\lrule B\wedge \negation P$ to express that $B$ must be false in every model.  A rule is \emph{safe} if each variable it mentions in the head occurs in some positive body
atom in a position other than a left operand of $\cal{S}$ or $\cal{U}$. A \emph{program} is a finite set of safe rules. A rule is \emph{ground} if it has no variables. For a program $\prog$, we denote by $\ground{\prog}$ the set of all the ground rules that can possibly be obtained from rules in $\prog$ by replacing variables with constants.

We conclude the preliminaries with the definitions of the \emph{two-valued semantics} of DatalogMTL$^\neg$, and the \emph{two valued models} of a dataset and a program.

\begin{definition}\label{def:twovaluedsemantics}
	Let $\prog$ be a program, $I$ an interpretation, and $t\in\timeline$. The \emph{two-valued semantics} of ground metric atoms and  bodies is defined as follows:
	\begin{enumerate}
		\item $\semtwo{\top}{I}{t}=\mtrue$,
		\item $\semtwo{\bot}{I}{t}=\mfalse$,
		\item $\semtwo{P(\mathbf{s})}{I}{t}=\begin{cases}
			\mtrue, & \text{if }P(\mathbf{s})\in I(t)
			\\ \mfalse, & \text{otherwise}
		\end{cases}$,
						\item $\semtwo{\diamondminus_\interval M}{I}{t}= \lub_{\leqtr}
					\{\semtwo{M}{I}{t'} \mid t'\in\timeline \text{ such that } t-t'\in \interval\}$,
					\item $\semtwo{\diamondplus_\interval M}{I}{t}= \lub_{\leqtr}
					\{\semtwo{M}{I}{t'} \mid t'\in\timeline \text{ such that } t'-t\in \interval\}$,
					\item $\semtwo{\boxminus_\interval M}{I}{t}=\glb_{\leqtr}
					\{\semtwo{M}{I}{t'} \mid t'\in\timeline \text{ such that } t-t'\in \interval\}$,
					\item $\semtwo{\boxplus_\interval M}{I}{t}= \glb_{\leqtr}
					\{\semtwo{M}{I}{t'} \mid t'\in\timeline \text{ such that } t'-t\in \interval\}$,
					\item $\semtwo{M_1 \mathcal{S}_\interval M_2}{I}{t}=\begin{cases}
						\mtrue, & \text{if } \semtwo{M_2}{I}{t'}=\mtrue \text{ for some }t' \text{ with } t-t'\in \interval 
						\\ & \text{ and } \semtwo{M_1}{I}{t''}=\mtrue \text{ for all }t''\in(t', t)
						\\ \mfalse, & \text{otherwise}
					\end{cases}$,
					\item $\semtwo{M_1 \mathcal{U}_\interval M_2}{I}{t}=\begin{cases}
						\mtrue, & \text{if } \semtwo{M_2}{I}{t'}=\mtrue \text{ for some }t' \text{ with } t'-t\in \interval 
						\\ & \text{ and } \semtwo{M_1}{I}{t''}=\mtrue \text{ for all }t''\in(t, t')
						\\ \mfalse, & \text{otherwise}
					\end{cases}$,
					\item $\semtwo{\negation M}{I}{t}=\semtwo{M}{I}{t}^{-1}$,
					\item $\semtwo{M_1\wedge \ldots \wedge M_m}{I}{t}=\glb_{\leqtr}\{\semtwo{M_1}{I}{t}, \ldots,\semtwo{M_m}{I}{t}\}$,
				\end{enumerate}
				where we define $\mfalse^{-1}=\mtrue$ and $\mtrue^{-1}=\mfalse$, and $\leqtr$ is the \emph{truth order}, i.e.\ the order defined by $\mfalse \leqtr \mtrue$. 
			\end{definition}

			\begin{definition}
				Let $\prog$ be a program, $\db$ a dataset, and $I\in \interpret$ an interpretation. Then, $I$ is a \emph{(two-valued) model} of $\prog$ iff for every rule $M \lrule B\in\ground{\prog}$, and for every $t\in\timeline$, $\semtwo{B}{I}{t}\leqtr\semtwo{M}{I}{t}$. $I$ is a \emph{(two-valued) model} of $\prog$ and $\db$ iff $I$ is a two-valued model of $\prog$ and a model of $\db$.
			\end{definition}

\section{The Three-valued Semantics of DatalogMTL$^\neg$ and the Immediate Consequence Operators}\label{sec:threeval}

In order to apply the machinery of AFT, we just need to define an approximation space and a suitable approximator. This turns out to be fairly straightforward: the approximation space is the space of \emph{three-valued interpretations}, and the approximator is an adaptation of Fitting's \emph{three valued immediate consequence operator} \cite{Fitting02}. 

In this section, we first define the three-valued interpretations and the three-valued semantics of rule bodies. Then, we  define the two-valued, and the three-valued immediate consequence operators, and we show they characterize the two-valued, and three valued models, respectively.

Defining the three-valued semantics of rule bodies can easily be done by using three-valued interpretations. A \emph{three-valued interpretation} $\mathcal{I}$ associates to each $t\in\timeline$ a function $\mathcal{I}(t)$ taking a ground relational atom and returning either $\mtrue$, $\mfalse$, or $\mundef$. Notice that three-valued interpretations may also be regarded as pairs of two-valued interpretations $(I_1,I_2)$, where $I_1, I_2\in\interpret$ such that $I_1\subseteq I_2$: an atom in $I_1(t)$ is interpreted as $\mtrue$, an atom in $I_2(t)$ but not in $I_1(t)$ as $\mundef$, and an atom not in $I_2(t)$ as $\mfalse$ (at time $t$). 
Notice that we have the following correspondence of truth values between the three-valued approach and the ``pairs'' approach: $\mtrue \equalhat(\mtrue, \mtrue) $, $\mundef \equalhat(\mfalse, \mtrue)$, and $\mfalse\equalhat(\mfalse, \mfalse)$. 
For the sake of clarity, we will use caligraphic fonts to differentiate three-valued interpretations from two-valued ones. For a three-valued interpretation $\mathcal{I}$, we will use the notation $I_1$ and $I_2$ for the two two-valued interpretations forming the pair $\mathcal{I}$ corresponds to., \ie $\mathcal{I}=(I_1,I_2)$. We will denote by $\intercons$ the set of three-valued interpretations, following the notation introduced in Section \ref{sec:prel:AFT}. We can now define the three-valued semantics  straightforwardly. Notice that we employ the same notation used for the two-valued semantics in Definition \ref{def:twovaluedsemantics}; the two semantics are easy to distinguish by looking at the interpretation (two- or three-valued) in the argument.

\begin{definition}\label{def:threevaluedsemantics}
	Let $\prog$ be a program, $\mathcal{I}$ a three-valued interpretation, and $t\in\timeline$. The \emph{three-valued semantics} of ground metric atoms and bodies is defined as follows:

		\begin{enumerate}
			\item $\semcal{M}{I}{t}=(\semtwo{M}{I_1}{t}, \semtwo{M}{I_2}{t})$, 
			\item $\semcal{\negation M}{I}{t}=\semcal{ M}{I}{t}^{-1}$,
			\item $\semcal{E_1\wedge \ldots \wedge E_m}{I}{t}= \glb_{\leqtr}\{\semcal{E_1}{I}{t}, \ldots, \semcal{ E_m}{I}{t}\}$,
				\end{enumerate}
		where we define $\mfalse^{-1}=\mtrue$, $\mundef^{-1}=\mundef$, and $\mtrue^{-1}=\mfalse$, and $\leqtr$ is the (three-valued) \emph{truth order}, i.e.\ the order defined by $\mfalse \leqtr \mundef \leqtr \mtrue$. 
	\end{definition}
	
	\begin{remark}\label{remark:threevaluedsemantics}
		Notice that in Definition \ref{def:threevaluedsemantics} it holds that 
		\begin{equation*}
			\semcal{E_1\wedge \ldots \wedge E_m}{I}{t}= \glb_{\leqtr}\{\semcal{E_1}{I}{t}, \ldots, \semcal{ E_m}{I}{t}\}=
			(\semtwo{E_1\wedge \ldots \wedge E_m}{I_1}{t},\semtwo{E_1\wedge \ldots \wedge E_m}{I_2}{t}).
		\end{equation*}
	\end{remark}

Analogously to what we did in Section \ref{sec:prel:datalog}, we naturally obtain the definition of \emph{three-valued model}.

\begin{definition}\label{def:threevalmodel}
	Let $\prog$ be a program, $\db$ a dataset, and $\mathcal{I}\in \intercons$ a three-valued interpretation. Then, $\mathcal{I}$ is a \emph{(three-valued) model} of $\prog$ iff for every rule $M \lrule B\in\ground{\prog}$, and for every $t\in\timeline$, $\semcal{B}{I}{t}\leqtr\semcal{M}{I}{t}$. $\mathcal{I}$ is a \emph{(three-valued) model} of $\prog$ and $\db$ iff $\mathcal{I}$ is a three-valued model of $\prog$ and $I_1$ is a model of $\db$.
\end{definition}

 Before defining the two-valued and the three-valued immediate consequence operators, we build a function $F$ to make explicit the link between a ground head atom and the time points it refers to.  Recall that we allow the nesting of $\boxplus$ and $\boxminus$ operators in the head of rules.

For all $t\in\timeline$, and for all ground head atoms $M\in\headatom$, we define:

\begin{equation*}
	F(M,t):=\begin{cases}
		\bigcup_{t'\in \timeline\mid t-t'\in\interval}\{F(M', t')\} & \text{ if } M=\boxminus_{\interval}M'
		\\ 	\bigcup_{t'\in \timeline\mid t'-t\in\interval}\{F(M', t')\} & \text{ if } M=\boxplus_{\interval}M'
		\\ \{(M,t)\} & \text{ otherwise}
	\end{cases}.
\end{equation*}

 From the definitions of the two-valued and the three-valued semantics, and the construction of $F$, it is easy to see the relation between  a ground head atom $M$ at time $t$ and  the relational facts it refers to, as expressed more precisely in the following propositions.

\begin{proposition}\label{prop:propertyF}
	Let  $M\in\headatom$, $\mathcal{I}\in \intercons$, and $t \in \timeline$. The following statements hold:
	\begin{enumerate}
		\item\label{item:two-valecase} $\semtwo{M}{I}{t}=\mtrue$ if and only if $\semtwo{P(\mathbf{s})}{I}{t'}=\mtrue$ for all $(P(\mathbf{s}), t')\in F(M,t)$.
		\item\label{item:three-valcase} $\semcal{M}{I}{t}=\glb_{\leqtr}\{\semcal{P(\mathbf{s})}{I}{t'}\mid(P(\mathbf{s}), t')\in F(M,t)\}$.
	\end{enumerate}
\end{proposition}
\begin{proof}
	Item \ref{item:two-valecase} is obvious by the definition of $F$. For Item \ref{item:three-valcase}, by Definition \ref{def:threevaluedsemantics}, we can write the following equalities: $\semcal{M}{I}{t}=(\semtwo{M}{I_1}{t}, \semtwo{M}{I_2}{t})$, and $\semcal{P(\mathbf{s})}{I}{t'}=(\semtwo{P(\mathbf{s})}{I_1}{t'},\semtwo{P(\mathbf{s})}{I_2}{t'})$ for all $(P(\mathbf{s}), t')\in F(M,t)$. Then Item \ref{item:three-valcase} follows easily from Item \ref{item:two-valecase}.
\end{proof}

Now, the definition of the two-valued, and three-valued immediate consequence operators for a database $\db$ and a program $\prog$ naturally follow.

\begin{definition}\label{def:immediateconsequenceop}
	Let $\prog$ be a program, and $\db$ be a dataset. The mapping $\immed{\db}{\prog}\colon\interpret\to\interpret$ is called the \emph{immediate consequence operator for} $\db$ \emph{and}  $\prog$, and it is defined for all $t\in \timeline$ as follows
	\begin{equation*}
		\begin{split}
				\immed{\db}{\prog}(I)(t):=&\{P(\mathbf{s})\mid \exists \interval \text{ such that } P(\mathbf{s})@\interval\in \db \text{ and } t\in\interval\}\\ &\cup \{P(\mathbf{s})\mid \exists M\!\lrule B\!\in\! \ground{\prog}, \exists t' \!\!\in\!\!\timeline\text{ such that } \semtwo{B}{I}{t'}=\mtrue \text{ and } (P(\mathbf{s}), t)\!\in\! F(M, t') \}
		\end{split}
	\end{equation*}
\end{definition}

\begin{definition} \label{def:threevalimmediate}
	Let $\prog$ be a program, and $\db$ be a dataset. The mapping $\appr{\db}{\prog}\colon\intercons\to\intercons$ is called the \emph{three-valued immediate consequence operator for} $\db$ \emph{and}  $\prog$, 
	and its components $\appr{\db}{\prog}^1\colon\intercons\to\interpret$, and $ \appr{\db}{\prog}^2\colon\intercons\to\interpret$ are defined for all $t\in \timeline$ as follows
	\begin{equation*}
		\begin{split}
			&	\begin{split}
				\appr{\db}{\prog}^1(\mathcal{I})(t):=&\{P(\mathbf{s})\mid \exists \interval \text{ such that } P(\mathbf{s})@\interval\in \db \text{ and } t\in\interval\}\\ 
				&\cup\! \{P(\mathbf{s})\mid \exists M\!\lrule\! B\!\in\! \ground{\prog}, \exists t'\! \!\!\in\!\!\timeline\text{ such that } \semcal{B}{I}{t'}\!=\!\mtrue \text{ and } (\!P(\!\mathbf{s}\!)\!,\! t\!)\!\in\! F(\!M\!,\! t'\!)\! \},
			\end{split}
			\\	
			& \begin{split}
				\appr{\db}{\prog}^2(\mathcal{I})(t):=&\{P(\mathbf{s})\mid \exists \interval \text{ such that } P(\mathbf{s})@\interval\in \db \text{ and } t\in\interval\}\\ 
				&\cup\! \{P(\mathbf{s})\mid \exists M\!\lrule\! B\!\in\! \ground{\prog}, \exists t'\!\!\! \!\in\!\!\timeline \text{ such that } \semcal{B}{I}{t'}\!\not=\!\mfalse \text{ and } (\!P(\!\mathbf{s}\!)\!,\! t\!)\!\in\!\! F(\!M\!,\! t'\!)\! \},
			\end{split}
		\end{split}
	\end{equation*}
\end{definition}

Finally, as  expected, $\immed{\db}{\prog}$ and $\appr{\db}{\Pi}$ characterize the two-valued, and the three-valued models, respectively.

\begin{proposition}
	Let $\prog$ be a program, $\db$ a database, and $I\in\interpret$. The interpretation $I$ is a two-valued model of $\db$ and $\prog$ if and only if $I$ is a pre-fixpoint of $\immed{\db}{\prog}$.
\end{proposition}
\begin{proof}
	Suppose that $I$ is a two-valued model of $\db$ and $\prog$. We have to show that for all $t\in\timeline$, $\immed{\db}{\prog}(I)(t)\subseteq I(t)$. Let $P(\mathbf{s})\in \immed{\db}{\prog}(I)(t)$. If there exists $\interval$ such that $t\in\interval$ and $P(\mathbf{s})@\interval\in \db$, then $P(\mathbf{s})\in I(t)$ because $I$ is, in particular, a model of $\db$. Let us now assume that such an interval $\interval$ does not exist. Then, it must be the case that there exists a ground rule $M\lrule B\in \ground{\prog}$ and a time point $t'\in\timeline$ such that $ \semtwo{B}{I}{t'}=\mtrue$ and $(P(\mathbf{s}), t)\in F(M, t')$. Since $I$ is a two-valued model of $\prog$, it holds that $\semtwo{B}{I}{t'}\leqtr\semtwo{M}{I}{t'}$; hence, it must hold that $\semtwo{M}{I}{t'}=\mtrue$. By Proposition \ref{prop:propertyF}, it follows that $\semtwo{P(\mathbf{s})}{I}{t}=\mtrue$, \ie $P(\mathbf{s})\in I(t)$, as desired.
	
	Conversely, assume that for all $t\in\timeline$, $\immed{\db}{\prog}(I)(t)\subseteq I(t)$. By definition of $\immed{\db}{\prog}$, for all ground relational facts $P(\mathbf{s})@\interval\in\db$, and for all $t\in\interval$, it holds that $P(\mathbf{s})\in \immed{\db}{\prog}(I)(t)\subseteq I(t)$. Hence, $I$ is a model of $\db$. We now prove that $I$ is also a two-valued model of $\prog$. Let $M\lrule B\in \ground{\prog}$, and $t\in\timeline$, and suppose $\semtwo{B}{I}{t}=\mtrue$. It suffices to show that $\semtwo{M}{I}{t}=\mtrue$. If $M=\top$, this is trivially verified. Suppose now that $M\neq \top$. Observe that, by definition of $\immed{\db}{\prog}$, for all $(P(\mathbf{s}),t')\in F(M,t)$, $P(\mathbf{s})\in \immed{\db}{\prog}(I)(t')\subseteq I(t')$. In particular, $\semtwo{P(\mathbf{s})}{I}{t'}=\mtrue$ for all $(P(\mathbf{s}),t')\in F(M,t)$. Hence, by Proposition \ref{prop:propertyF}, $\semtwo{M}{I}{t}=\mtrue$, as desired.
\end{proof}

\begin{proposition}\label{prop:threevaliffprefixpt}
	Let $\prog$ be a program, $\db$ a database, and $\mathcal{I}\in\intercons$. The interpretation $\mathcal{I}$ is a three-valued model of $\db$ and $\prog$ if and only if $\mathcal{I}$ is a pre-fixpoint of $\appr{\db}{\prog}$.
\end{proposition}
\begin{proof}
		Suppose that $\mathcal{I}$ is a three-valued model of $\db$ and $\prog$. We have to show that $\appr{\db}{\prog}(\mathcal{I})\leqtr \mathcal{I}$. Let $P(\mathbf{s})$ be a ground relational atom such that $ \appr{\db}{\prog}(\mathcal{I})(t)(P(\mathbf{s}))=\mtrue$. Hence, by Definition \ref{def:threevalimmediate}, either there exists an interval $\interval$ such that $t\in\interval$ and $P(\mathbf{s})@\interval\in \db$, in which case $\mathcal{I}(P(\mathbf{s}))=\mtrue$ (as $\mathcal{I}$ is a model of $\db$ by hypothesis); or there exists a ground rule $M\lrule B\in \ground{\prog}$ and a time point $t'\in\timeline$ such that $ \semcal{B}{\mathcal{I}}{t'}=\mtrue$ and $(P(\mathbf{s}), t)\in F(M, t')$. Since $\mathcal{I}$ is a three-valued model of $\prog$, it holds that $\semcal{B}{\mathcal{I}}{t'}\leqtr\semcal{M}{\mathcal{I}}{t'}$; hence, it must hold that $\semcal{M}{\mathcal{I}}{t'}=\mtrue$. By Proposition \ref{prop:propertyF}, it follows that $\mathcal{I}(t)(P(\mathbf{s}))=(I^1(t)(P(\mathbf{s})),I^2(t)(P(\mathbf{s}))=(\semtwo{P(\mathbf{s})}{I^1}{t},\semtwo{P(\mathbf{s})}{I^2}{t})=\semcal{P(\mathbf{s})}{I}{t}=\mtrue$,  as desired. Now let $P(\mathbf{s})$ be a ground relational atom such that $ \appr{\db}{\prog}(\mathcal{I})(t)(P(\mathbf{s}))=\mundef$. Hence, by Definition \ref{def:threevalimmediate}, there exists a ground rule $M\lrule B\in \ground{\prog}$ and a time point $t'\in\timeline$ such that $ \semcal{B}{\mathcal{I}}{t'}\neq\mfalse$ and $(P(\mathbf{s}), t)\in F(M, t')$. Since $\mathcal{I}$ is a three-valued model of $\prog$, it holds that $\mundef\leqtr\semcal{B}{\mathcal{I}}{t'}\leqtr\semcal{M}{\mathcal{I}}{t'}$. In particular, we have that $\semcal{M}{\mathcal{I}}{t'}\neq\mfalse$. By Proposition \ref{prop:propertyF}, it follows that $\mathcal{I}(t)(P(\mathbf{s}))=\semcal{P(\mathbf{s})}{I}{t}=\mfalse$, concluding the first half of the proof.

	Conversely, assume that for all $t\in\timeline$, $\appr{\db}{\prog}(\mathcal{I})(t)\leqtr \mathcal{I}(t)$. By definition of $\appr{\db}{\prog}$, for all ground relational facts $P(\mathbf{s})@\interval\in\db$, and for all $t\in\interval$, it holds that $ \mtrue=\appr{\db}{\prog}(\mathcal{I})(t)(P(\mathbf{s}))\leqtr \mathcal{I}(t)(P(\mathbf{s}))$. Hence, $\mathcal{I}$ is a model of $\db$. We now prove that $\mathcal{I}$ is also a three-valued model of $\prog$. Let $M\lrule B\in \ground{\prog}$, and $t\in\timeline$. We have to show that $\semcal{B}{I}{t}\leqtr\semcal{M}{I}{t}$. If $M=\top$, this is trivially verified, so assume that $M\neq \top$. We first suppose $\semcal{B}{I}{t}=\mtrue$ and show that $\semcal{M}{I}{t}=\mtrue$. Observe that, by definition of $\appr{\db}{\prog}$, for all $(Q(\mathbf{s}),t')\in F(M,t)$, it holds that $ \mtrue=\appr{\db}{\prog}(\mathcal{I})(t')(Q(\mathbf{s}))\leqtr \mathcal{I}(t')(Q(\mathbf{s}))$. In particular, $\semcal{Q(\mathbf{s})}{I}{t'}=\mtrue$ for all $(Q(\mathbf{s}),t')\in F(M,t)$. Hence, by Proposition \ref{prop:propertyF}, $\semcal{M}{I}{t}=\mtrue$, as desired.
	Now, we assume $\semcal{B}{I}{t}=\mundef$ and show that $\semcal{M}{I}{t}\neq\mfalse$. Observe that, by definition of $\appr{\db}{\prog}$, for all $(Q(\mathbf{s}),t')\in F(M,t)$, it holds that $ \mundef\leqtr\appr{\db}{\prog}(\mathcal{I})(t')(Q(\mathbf{s}))\leqtr \mathcal{I}(t')(Q(\mathbf{s}))$. In particular, $\semcal{Q(\mathbf{s})}{I}{t'}\geqtr\mundef$ for all $(Q(\mathbf{s}),t')\in F(M,t)$. Hence, by Proposition \ref{prop:propertyF}, $\semcal{M}{I}{t}\geqtr\mundef$, concluding the proof. 
\end{proof}

\section{Approximation Fixpoint Theory for DatalogMTL$^\neg$}\label{sec:AFTdatalog}

After having defined the three-valued semantics and the three-valued immediate consequence operator in the previous section, it only remains to show that the three-valued immediate consequence operator $\appr{\db}{\prog}$ is a consistent approximator of $\immed{\db}{\prog}$ (Theorem \ref{thm:consistentapprox}). Then, the machinery of AFT allows for the definition of four kinds of models of a database $\db$ and a program $\prog$: \emph{supported}, \emph{stable}, \emph{Kripke-Kleene}, and \emph{well-founded}. 

 We start by proving that the space of two-valued interpretations $\interpret$ is a complete lattice, which follows straightforwardly from results of order theory.

\begin{proposition}\label{prop:lattice}
	$\langle \interpret, \subseteq \rangle$ is a complete lattice.
\end{proposition}
\begin{proof}
	Let $\mathcal{G}$ be the set of ground relational atoms. Then $\interpret$ is the set of functions with signature $\timeline \to 2^\mathcal{G}$. The power set $2^\mathcal{G}$ ordered with set inclusion $\subseteq$ is clearly a complete lattice, with bottom element the empty set, and top element $\mathcal{G}$. Moreover, $\timeline$ can be considered as a poset with the standard order on integers or rationals. Since the category of complete lattices is a full subcategory of the category of posets, by Proposition 2 in \cite{PollaciKDB2025}, we have the isomorphism $\interpret\cong \Pi_{t\in\timeline}2^\mathcal{G}$. By Proposition 1 in \cite{PollaciKDB2025}, $\Pi_{t\in\timeline}2^\mathcal{G}$ is a complete lattice, hence, $\interpret$ is a complete lattice too.
\end{proof}

Using the terminology of AFT, the set of three-valued interpretations $\intercons$ equipped with the precision order $\leq_p$ ($\mathcal{I}\leq_p \mathcal{J}$ iff $I_1\subseteq J_1\subseteq J_2\subseteq I_2$) is an \emph{approximation space} for $\interpret$. In order to prove Theorem \ref{thm:consistentapprox}, we make use of the following two lemmata, which ensure the evaluation of ground metric atoms preserves the order relation in both the two-valued and the three-valued case.

\begin{lemma}\label{lemma:twomonotonicity}
	Let $I_1, I_2\in \interpret$ be two-valued interpretations, and $M$ be a ground atom. If $I_1 \subseteq I_2$, then for all $t\in\timeline$, $\semtwo{M}{I_1}{t}\leqtr \semtwo{M}{I_2}{t}$. 
\end{lemma}
\begin{proof}
	Let $I_1\subseteq I_2$, \ie $I_1(t)\subseteq I_2(t)$ for all $t\in\timeline$.
	We proceed by induction on the structure of $M$.
	
	If $M\in\{\bot, \top\}$, then clearly $\semtwo{M}{I_1}{t}=\semtwo{M}{I_2}{t}$. 
	If $M=P(\mathbf{s})$, then $\semtwo{M}{I_1}{t}\leqtr \semtwo{M}{I_2}{t}$ follows from Definition \ref{def:twovaluedsemantics} and the assumption that $I_1(t)\subseteq I_2(t)$.
	
	Now, let $M= \diamondminus_\interval M'$, and assume $\semtwo{M'}{I_1}{t}\leqtr \semtwo{M'}{I_2}{t}$ for all $t\in\timeline$. If  $\semtwo{M}{I_1}{t}=\mtrue$ for some $t\in\timeline$, then by Definition \ref{def:twovaluedsemantics}, there exists $t'\in\timeline$ such that $t-t'\in\interval$ for which $\semtwo{M'}{I_1}{t'}=\mtrue$. By induction hypothesis, it must hold that  $\semtwo{M'}{I_2}{t'}=\mtrue$, as desired. 
	We can proceed analogously for the cases $M= \diamondplus_\interval M'$, $M= \boxminus_\interval M'$, and $M= \boxplus_\interval M'$.
	
	Let $M= M_1 \mathcal{S}_\interval M_2$, and assume $\semtwo{M_1}{I_1}{t}\leqtr \semtwo{M_1}{I_2}{t}$, and $\semtwo{M_2}{I_1}{t}\leqtr \semtwo{M_2}{I_2}{t}$ for all $t\in\timeline$.
	If  $\semthree{M}{I_1}{J_1}{t}=\mtrue$ for some $t\in\timeline$, then by Definition \ref{def:twovaluedsemantics}, there exists $t'\in\timeline$ such that $t-t'\in\interval$,  $\semtwo{M_2}{I_1}{t'}=\mtrue$, and $\semtwo{M_1}{I_1}{t''}=\mtrue$ for all $t''\in(t',t)$. By induction hypothesis, it must hold that  $\semtwo{M_2}{I_2}{t'}=\mtrue$, and $\semtwo{M_1}{I_2}{t''}=\mtrue$ for all $t''\in(t',t)$. It follows that $\semtwo{M}{I_2}{t}=\mtrue$, as desired. 
	We can proceed analogously for the  case $M= M_1 \mathcal{U}_\interval M_2$.
\end{proof}

\begin{lemma}\label{lemma:monotonicity}
	Let $\mathcal{I}, \mathcal{J}\in \intercons$ be three-valued interpretations, and $E$ be a ground expression. If $\mathcal{I}\leq_p\mathcal{J}$, then for all $t\in\timeline$, $\semcal{E}{I}{t}\leq_p \semcal{E}{J}{t}$.
\end{lemma}
\begin{proof}
	Assume that $\mathcal{I}\leq_p\mathcal{I}$, \ie $I_1(t)\subseteq J_1(t) \subseteq J_2(t)\subseteq I_2(t)$ for all $t\in\timeline$.
	We proceed by induction on the structure of $E$.
	
	If $E$ is  a ground metric atom, then by Lemma \ref{lemma:twomonotonicity}, we have $\semtwo{E}{I_1}{t}\leqtr \semtwo{E}{J_1}{t}\leqtr \semtwo{E}{J_2}{t}\leqtr \semtwo{E}{I_2}{t}$. By Definition \ref{def:threevaluedsemantics}, we obtain $\semcal{E}{I}{t}\leq_p \semcal{E}{J}{t}$.

	Now assume $E=\negation M$  for some ground metric atom $M$, and that $\semcal{M}{I}{t}\leq_p \semcal{M}{J}{t}$.
If $\semcal{M}{I}{t}= \semcal{M}{J}{t}$, then clearly $\semcal{E}{I}{t}= \semcal{E}{J}{t}$. Now assume $\semcal{M}{I}{t}<_p \semcal{M}{J}{t}$. It must be that $\semcal{M}{I}{t}=\mundef$ and $\semcal{M}{J}{t}\in\{\mtrue, \mfalse\}$.  By Definition \ref{def:threevaluedsemantics}, it clearly follows that $\semcal{E}{I}{t}=\mundef$ and $\semcal{E}{J}{t}\in\{\mtrue, \mfalse\}$; hence $\semcal{E}{I}{t}\leq_p \semcal{E}{J}{t}$.
	
	Finally, for the case $E=E_1\wedge \ldots \wedge E_m$, we can easily use the induction hypothesis on $E_1, \ldots, E_m$ to obtain that $\semcal{E}{I}{t}\leq_p \semcal{E}{J}{t}$.
\end{proof}

\begin{theorem}\label{thm:consistentapprox}
	Let $\prog$ be a program, and $\db$ be a dataset. The operator $\appr{\db}{\prog}$ is a consistent approximator of $\immed{\db}{\prog}$. 
\end{theorem}
\begin{proof}
	We have to prove three things:
	\begin{enumerate}
		\item\label{item:consistent} \emph{$\appr{\db}{\prog}$ is consistent}, \ie for all $\mathcal{I}\in \intercons$, $\appr{\db}{\prog}(\mathcal{I})\in\intercons$.
		\item\label{item:monotone} \emph{$\appr{\db}{\prog}$ is $\leq_p$-monotone,} \ie for all $\mathcal{I}, \mathcal{J}\in \intercons$ such that $\mathcal{I}\leq_p \mathcal{J}$, it holds that $\appr{\db}{\prog}(\mathcal{I})\leq_p \appr{\db}{\prog}(\mathcal{J})$.
		\item\label{item:approximator} \emph{$\appr{\db}{\prog}$ approximates $\immed{\db}{\prog}$,} \ie for all $I\in \interpret$, $\appr{\db}{\prog}(I,I)=(\immed{\db}{\prog}(I),\immed{\db}{\prog}(I))$.
	\end{enumerate}
For Item \ref{item:consistent}, observe that from Definition \ref{def:threevalimmediate} we clearly have $	\appr{\db}{\prog}^1(\mathcal{I})(t)\subseteq 	\appr{\db}{\prog}^2(\mathcal{I})(t)$ for all $\mathcal{I}\in\intercons$ and for all $t\in\timeline$. Hence, for all $\mathcal{I}\in \intercons$, it holds that $\appr{\db}{\prog}(\mathcal{I})\in\intercons$, as desired.

Now, to prove Item \ref{item:monotone}, let $\mathcal{I}\leq_p \mathcal{J}$ be two elements of $\intercons$, and $t\in\timeline$. We need to show that $\appr{\db}{\prog}^1(\mathcal{I})(t)\subseteq \appr{\db}{\prog}^1(\mathcal{J})(t)$ and $\appr{\db}{\prog}^2(\mathcal{J})(t)\subseteq\appr{\db}{\prog}^2(\mathcal{I})(t)$. Let $P(\mathbf{s})\in \appr{\db}{\prog}^1(\mathcal{I})(t)$ (if there is no such $P(\mathbf{s})$, then we already have $\appr{\db}{\prog}^1(\mathcal{I})(t)\subseteq \appr{\db}{\prog}^1(\mathcal{J})(t)$).
 If there exists $\interval$ such that $t\in\interval$ and $ P(\mathbf{s}) @ \interval\in\db$, then $P(\mathbf{s})\in \appr{\db}{\prog}^1(\mathcal{J})(t)$ too. If this is not the case, then there exists a ground rule $M\lrule B\in\ground{\Pi}$ and $t'\in\timeline$ such that $\semcal{B}{I}{t'}=\mtrue$ and $(P(\mathbf{s}),t)\in F(M,t')$. By Lemma \ref{lemma:monotonicity}, it holds that $\semcal{B}{J}{t'}=\mtrue$ , hence $P(\mathbf{s})\in \appr{\db}{\prog}^1(\mathcal{J})(t)$. We have so far proved the first inclusion, \ie $\appr{\db}{\prog}^1(\mathcal{I})(t)\subseteq \appr{\db}{\prog}^1(\mathcal{J})(t)$. Showing that the second inclusion holds can be performed analogously.

Finally, for Item \ref{item:approximator}, let $B$ be the body of a ground rule in $\ground{\Pi}$, and $t\in\timeline$. It is easy to prove by induction on $B$, and using Definition \ref{def:threevaluedsemantics} and Remark \ref{remark:threevaluedsemantics}, that  $\semtwo{B}{(I,I)}{t}=(\semtwo{B}{I}{t},\semtwo{B}{I}{t})$. In particular, $\semtwo{B}{(I,I)}{t}\neq \mundef$, and $\semtwo{B}{(I,I)}{t}=\mtrue$  if and only if $\semtwo{B}{(I,I)}{t}\neq \mfalse$. Hence, by Definitions \ref{def:immediateconsequenceop} and \ref{def:threevalimmediate}, for all $t\in\timeline$, $	\appr{\db}{\prog}^1(I,I)(t)=	\appr{\db}{\prog}^2(I,I)(t)= \immed{\db}{\Pi}(I)(t)$, as desired.
\end{proof}

Thanks to Theorem \ref{thm:consistentapprox}, and the machinery of AFT, we are able to define the following models of interest.

\begin{definition}\label{def:models}
	Let $\db$ be a dataset, $\prog$ be a program, and $\mathcal{I}\in\intercons$. 
	\begin{itemize}
		\item $\mathcal{I}$ is a \emph{three-valued supported model of $\db$ and $\prog$} if it is a fixpoint of $\appr{\db}{\prog}$;
		\item $\mathcal{I}$ is a \emph{three-valued stable model of $\db$ and $\prog$} if it is a stable fixpoint of $\appr{\db}{\prog}$; that is, if $I_1 = \lfp \appr{\db}{\prog}^1(\cdot, I_2)$ and $I_2 = \lfp \appr{\db}{\prog}^2(I_1, \cdot)$, where $\appr{\db}{\prog}^1(\cdot, I_2)$ is the function that maps
		an interpretation $X$ to the first component of $\appr{\db}{\prog}(X,J)$, and similarly for $\appr{\db}{\prog}^2(I_1,\cdot)$;
		\item$\mathcal{I}$ is the \emph{Kripke-Kleene model of $\db$ and $\prog$} if it is the $\leq_p$-least fixpoint of $\appr{\db}{\prog}$;
		\item $\mathcal{I}$is the \emph{well-founded model of $\db$ and $\prog$} if it is the well-founded fixpoint of $\appr{\db}{\prog}$, \ie if it
		is the $\leq_p$-least three-valued stable model.
	\end{itemize}
We call a two-valued model $I$ of $\db$ and $\prog$ a \emph{two-valued stable  model of $\db$ and $\prog$} if $(I,I)$ is a three-valued stable model of $\db$ and $\prog$.
\end{definition}

\section{Comparison with Previous Work}\label{sec:comparison}

In this section, we compare the definitions of \emph{HT-models} and stable models of \cite{WalegaCGK24} with ours. In particular, we first show that any three-valued model is an HT-model from \cite{WalegaCGK24}, while the converse does not hold in general. Finally, we prove that indeed the notion of \emph{stable model} from \cite{WalegaCGK24} coincides with our definition of  two-valued stable model. For the sake of clarity, we rewrite the definitions of HT-model (Definition \ref{def:HTmodel}) and stable model (Definition \ref{def:stablewalega}) from \cite{WalegaCGK24} with the notation used in this paper. To avoid ambiguity, we will rename the stable models from \cite{WalegaCGK24} ``stable HT-models''.

\begin{definition}\label{def:HTmodel}
	Let $\db$ be a dataset, $\prog$ be a program, and $\mathcal{I}\in\intercons$. $\mathcal{I}$ is an \emph{HT-model of $\db$ and $\prog$} if $I_1$ is a model of $\db$ and for all rules $	M \lrule M_1\wedge\ldots \wedge M_k\wedge \negation M_{k+1}\wedge \ldots \wedge \negation M_m\in \ground{\prog}$, for all $t\in\timeline$, the following conditions hold
	\begin{enumerate}
		\item\label{item:HTtrue} if $\semtwo{M_i}{I_1}{t}=\mtrue$ for all $i\in\{1,\ldots, k\}$ and $\semtwo{M_j}{I_2}{t}=\mfalse$ for all $j\in\{k+1,\ldots, m\}$, then $\semtwo{M}{I_1}{t}=\mtrue$;
		\item\label{item:HTundef} if $\semtwo{M_i}{I_2}{t}=\mtrue$ for all $i\in\{1,\ldots, k\}$ and $\semtwo{M_j}{I_2}{t}=\mfalse$ for all $j\in\{k+1,\ldots, m\}$, then $\semtwo{M}{I_2}{t}=\mtrue$.
	\end{enumerate}
\end{definition}

\begin{proposition}\label{prop:HTiffthreevalued}
	Let $\db$ be a dataset, $\prog$ be a program, and $\mathcal{I}\in\intercons$. If $\mathcal{I}$ is a three-valued model of $\db$ and $\prog$, then $\mathcal{I}$ is an HT-model of $\db$ and $\prog$.
\end{proposition}
\begin{proof}
	Clearly, $I_1$ is a model for $\db$ by hypothesis. We now show that Item \ref{item:HTtrue} from Definition \ref{def:HTmodel} holds.
	Let $M \lrule B\in\ground{\prog}$ be a ground rule with $B=M_1\wedge\ldots \wedge M_k\wedge \negation M_{k+1}\wedge \ldots \wedge \negation M_m$, and $t\in\timeline$ such that  $\semtwo{M_i}{I_1}{t}=\mtrue$ for all $i\in\{1,\ldots, k\}$ and $\semtwo{M_j}{I_2}{t}=\mfalse$ for all $j\in\{k+1,\ldots, m\}$. This implies that $\semcal{M_i}{I}{t}=\mtrue$ for all $i\in\{1,\ldots, k\}$ and $\semcal{M_j}{I}{t}=\mfalse$ for all $j\in\{k+1,\ldots, m\}$. In particular, by Defintion \ref{def:threevaluedsemantics} we have that $\semcal{B}{I}{t}=\mtrue$. Since $\mathcal{I}$ is a three-valued model by hypothesis, it must be that $\semcal{M}{I}{t}=\mtrue$ too. Hence, $\semtwo{M}{I_1}{t}=\mtrue$, as desired.
	
	Now, we prove that Item \ref{item:HTundef} from Definition \ref{def:HTmodel} holds as well.
	Let $M \lrule B\in\ground{\prog}$ be a ground rule with $B=M_1\wedge\ldots \wedge M_k\wedge \negation M_{k+1}\wedge \ldots \wedge \negation M_m$, and $t\in\timeline$ such that  $\semtwo{M_i}{I_2}{t}=\mtrue$ for all $i\in\{1,\ldots, k\}$ and $\semtwo{M_j}{I_2}{t}=\mfalse$ for all $j\in\{k+1,\ldots, m\}$. This implies that $\semcal{M_i}{I}{t}\neq \mfalse$ for all $i\in\{1,\ldots, k\}$ and $\semcal{M_j}{I}{t}=\mfalse$ for all $j\in\{k+1,\ldots, m\}$. In particular, by Definition \ref{def:threevaluedsemantics} we have that $\semcal{B}{I}{t}\neq \mfalse$. Since $\mathcal{I}$ is a three-valued model, it must be that $\semcal{M}{I}{t}\neq\mfalse$. Hence, $\semtwo{M}{I_2}{t}=\mtrue$, concluding the proof.
\end{proof}

The converse of Proposition \ref{prop:HTiffthreevalued} does not hold, i.e.\ not all HT-models are three-valued models. In fact, by Definition \ref{def:threevalmodel}, $\mathcal{I}\in \intercons$  is a three-valued model of $\prog$ iff for every rule $M \lrule B\in\ground{\prog}$, and for every $t\in\timeline$, it holds that $\semcal{B}{I}{t}\leqtr\semcal{M}{I}{t}$, i.e.\ it holds that
	\begin{enumerate}
	\item\label{item:threetrue} if $\semtwo{M_i}{I_1}{t}=\mtrue$ for all $i\in\{1,\ldots, k\}$ and $\semtwo{M_j}{I_2}{t}=\mfalse$ for all $j\in\{k+1,\ldots, m\}$, then $\semtwo{M}{I_1}{t}=\mtrue$,
	\item\label{item:threeundef} if $\semtwo{M_i}{I_2}{t}=\mtrue$ for all $i\in\{1,\ldots, k\}$ and $\semtwo{M_j}{I_1}{t}=\mfalse$ for all $j\in\{k+1,\ldots, m\}$, then $\semtwo{M}{I_2}{t}=\mtrue$,
\end{enumerate}
where $B:=M_1\wedge\ldots \wedge M_k\wedge \negation M_{k+1}\wedge \ldots \wedge \negation M_m$. Notice that, while the first condition above is the same as the first one in Definition \ref{def:HTmodel}, the second condition is more strict than the second condition of Definition \ref{def:HTmodel}: it applies even when $\semtwo{M_h}{I_2}{t}=\mtrue$ for some $h\in\{k+1,\ldots, m\}$, provided that $\semtwo{M_j}{I_1}{t}=\mfalse$ for all $j\in\{k+1,\ldots, m\}$. 

We show this discrepancy more concretely in the example below.

\begin{example}
	Consider an empty dataset $\db$ and a program $\Pi$ with just one rule, namely $P\lrule \negation Q$, where $P$ and $Q$ are predicates of arity $0$. Let $(I_1,I_2)$ be the three-valued interpretation with $I_1(t)=\emptyset$ and $I_2(t)=\{Q\}$ for all $t\in\timeline$. Clearly, $I_1$ is a model for the empty dataset $\db$. Moreover, since the antecedent in both conditions of Definition \ref{def:HTmodel} are not satisfied, $(I_1,I_2)$ is an HT-model for $\db$ and $\Pi$. However, $(I_1,I_2)$ is not a three-valued model: $\semthree{\negation Q}{I_1}{I_2}{t}=\semthree{Q}{I_1}{I_2}{t}^{-1}=(\semtwo{Q}{I_1}{t}, \semtwo{Q}{I_2}{t})^{-1}=\mundef^{-1}=\mundef \not\leqtr \mfalse =\semthree{P}{I_1}{I_2}{t}$, for all $t\in\timeline$.
\end{example}

We now focus on \emph{stable HT-models}, and show that they correspond to our two-valued stable models.

\begin{definition}\label{def:stablewalega}
	Let $\db$ be a dataset, $\prog$ be a program, and $I\in\interpret$. $I$ is a \emph{stable HT-model of $\db$ and $\prog$} if $(I,I)$ is an HT-model of $\db$ and $\prog$ and 
	\begin{equation}\label{eq:glbWalegastable}
		I=\glb_\subseteq\{J\mid (J,I) \text{ is an HT-model of }\db\text{ and }\prog\}.
	\end{equation}
\end{definition}

\begin{lemma}\label{lemma:HTprefixpt}
	Let $\db$ be a dataset, $\prog$ be a program, and $I, J\in\interpret$. If $(I,J)$ is an HT-model, then $I$ is a prefixpoint of $\appr{\db}{\Pi}^1(\cdot, J)$.
\end{lemma}
\begin{proof}
	We have to show that $\appr{\db}{\Pi}^1(I, J)\subseteq I$. Let $t\in\timeline$, and $P(\mathbf{s})\in \appr{\db}{\Pi}^1(I,J)(t)$. Then, there are two possible cases:
	\begin{enumerate}
		\item There exists an interval $\delta$ such that $t\in\delta$ and $P(\mathbf{s})@ \delta \in\db$. In this case, it must hold that $P(\mathbf{s})\in I(t)$, because $(I,J)$ is an HT-model of $\db$ and $\Pi$.
		\item There exist $M\lrule B\in\ground{\Pi}$ and $t'\in\timeline$ such that $\semthree{B}{I}{J}{t'}=\mtrue$ and $(P(\mathbf{s}),t)\in F(M,t')$. By Definition \ref{def:threevaluedsemantics}, it is easy to see that $\semthree{B}{I}{J}{t'}=\mtrue$ satisfies the antecedent of Item \ref{item:HTtrue} of Definition \ref{def:HTmodel}. Since $(I,J)$ is an HT-model, we can derive that $\semtwo{M}{I}{t'}=\mtrue$. By Proposition \ref{prop:propertyF}, we have that $\semtwo{P(\mathbf{s})}{I}{t}=\mtrue$, i.e.\ $P(\mathbf{s})\in I(t)$.
	\end{enumerate} 
	In both cases above we can conclude that $P(\mathbf{s})\in I(t)$. By the arbitrarity of $P(\mathbf{s})$ and $t$, it holds that $\appr{\db}{\Pi}^1(I,J)\subseteq I$, i.e.\ $I$ is a prefixpoint of $\appr{\db}{\Pi}^1(\cdot,J)$, as desired.
\end{proof}

\begin{theorem}
	Let $\db$ be a dataset, $\prog$ be a program, and $I\in\interpret$. $I$ is  a stable HT-model of $\db$ and $\Pi$ if and only if $I$ is a two-valued stable model of $\db$ and $\Pi$. 
\end{theorem}
\begin{proof}
	Suppose $I$ is a stable HT-model. By Definition \ref{def:HTmodel}, $I$ is a model of $\db$. We now prove that $(I,I)$ is a stable fixpoint of $\appr{\db}{\Pi}$, \ie $I=\lfp \appr{\db}{\Pi}^1(\cdot,I)=\lfp\appr{\db}{\Pi}^2(I,\cdot)$. By Theorem \ref{thm:consistentapprox},  $\appr{\db}{\Pi}^1(I,I)=I$ if and only if $\appr{\db}{\Pi}^2(I,I)=I$, \ie $I$ is a fixpoint of  $\appr{\db}{\Pi}^1(\cdot,I)$ if and only if $I$ is a fixpoint of  $\appr{\db}{\Pi}^2(I,\cdot)$. 
	Since we consider only consistent pairs, \ie pairs $(J_1,J_2)$ with $J_1\subseteq J_2$, it is sufficient to prove that   $I=\lfp \appr{\db}{\Pi}^1(\cdot,I)$. 
	 By Lemma \ref{lemma:HTprefixpt}, $I$ is a pre-fixpoint of  $\appr{\db}{\Pi}^1(\cdot,I)$. We now show that $I\subseteq \appr{\db}{\Pi}^1(I,I)$. Since $I$ is a stable HT-model, it suffices to  prove that $(\appr{\db}{\Pi}^1(I,I),I)$ is an HT-model of $\db$ and $\Pi$. Clearly, $\appr{\db}{\Pi}^1(I,I)$ is a model of $\db$ by Definition \ref{def:threevalimmediate}. It remains to verify the two conditions listed in Definition \ref{def:HTmodel}. Let $B=M_1\wedge\ldots \wedge M_k\wedge \negation M_{k+1}\wedge \ldots \wedge \negation M_m$ such that $	M \lrule B \in \ground{\prog}$, and $t\in\timeline$:
	 \begin{enumerate}
	 	\item Suppose $\semtwo{M_i}{\appr{\db}{\Pi}^1(I,I)}{t}=\mtrue$ for all $i\in\{1,\ldots, k\}$ and $\semtwo{M_j}{I}{t}=\mfalse$ for all $j\in\{k+1,\ldots, m\}$. By Lemma \ref{lemma:twomonotonicity}, it holds that $\semtwo{M_i}{I}{t}=\mtrue$ for all $i\in\{1,\ldots, k\}$. Clearly, since all the $M_i$'s do not contain negation, we also have that $\semthree{M_i}{I}{I}{t}=\mtrue$ for all $i\in\{1,\ldots, k\}$ and $\semthree{M_j}{I}{I}{t}=\mfalse$ for all $j\in\{k+1,\ldots, m\}$. Hence, we have that $\semthree{B}{I}{I}{t}=\mtrue$. By Definition \ref{def:threevalimmediate}, $P(\mathbf{s})\in\appr{\db}{\Pi}^1(I,I)(t')$ for all $t'\in\timeline$ and $P(\mathbf{s})$ ground relational atoms such that $(P(\mathbf{s}),t')\in F(M,t)$. By Proposition \ref{prop:propertyF}, it follows that $\semtwo{M}{\appr{\db}{\Pi}^1(I,I)}{t}=\mtrue$, as desired.
	 	\item Suppose $\semtwo{M_i}{I}{t}=\mtrue$ for all $i\in\{1,\ldots, k\}$ and $\semtwo{M_j}{I}{t}=\mfalse$ for all $j\in\{k+1,\ldots, m\}$. By the same reasoning performed at the point above, it follows that $\semtwo{M}{\appr{\db}{\Pi}^1(I,I)}{t}=\mtrue$. By Lemma \ref{lemma:twomonotonicity}, we get that $\semtwo{M}{I}{t}=\mtrue$, as desired
	 \end{enumerate}
 Thus, we can conclude that $(\appr{\db}{\Pi}^1(I,I),I)$ is an HT-model. Since $I$ is a stable HT-model, it follows that $\appr{\db}{\Pi}^1(I,I)=I$, i.e.\ $I$ is a fixpoint of $ \appr{\db}{\Pi}^1(\cdot,I)$. 
 
 We prove now that $I$ is the \emph{least} fixpoint of $\appr{\db}{\Pi}^1(\cdot,I)$.  Let $J\subseteq I$ be a fixpoint of  $\appr{\db}{\Pi}^1(\cdot,I)$. It suffices to show that $I\subseteq J$. If we show that $(J,I)$ is an HT-model, we can again conclude by the minimality of $I$. Hence, we proceed to show that $(J,I)$ is indeed an HT-model of $\db$ and $\Pi$. First, observe that since $J=\appr{\db}{\Pi}^1(J,I)$, $J$ is a model of $\db$ by Definition \ref{def:threevalimmediate}. It remains to verify the two conditions listed in Definition \ref{def:HTmodel}. Let $B=M_1\wedge\ldots \wedge M_k\wedge \negation M_{k+1}\wedge \ldots \wedge \negation M_m$ such that $	M \lrule B \in \ground{\prog}$, and $t\in\timeline$:
 \begin{enumerate}
 	\item Suppose $\semtwo{M_i}{J}{t}=\mtrue$ for all $i\in\{1,\ldots, k\}$ and $\semtwo{M_j}{I}{t}=\mfalse$ for all $j\in\{k+1,\ldots, m\}$. By Lemma \ref{lemma:twomonotonicity}, it holds that $\semtwo{M_i}{I}{t}=\mtrue$ for all $i\in\{1,\ldots, k\}$, and $\semtwo{M_j}{J}{t}=\mfalse$ for all $j\in\{k+1,\ldots, m\}$. Hence, we have that $\semthree{M_i}{J}{I}{t}=\mtrue$ for all $i\in\{1,\ldots, k\}$ and $\semthree{M_j}{J}{I}{t}=\mfalse$ for all $j\in\{k+1,\ldots, m\}$. It follows that $\semthree{B}{J}{I}{t}=\mtrue$. By Definition \ref{def:threevalimmediate}, $P(\mathbf{s})\in\appr{\db}{\Pi}^1(J,I)(t')=J(t')$ for all $t'\in\timeline$ and $P(\mathbf{s})$ ground relational atoms such that $(P(\mathbf{s}),t')\in F(M,t)$. By Proposition \ref{prop:propertyF}, it follows that $\semtwo{M}{J}{t}=\mtrue$, as desired.
 	\item Suppose $\semtwo{M_i}{I}{t}=\mtrue$ for all $i\in\{1,\ldots, k\}$ and $\semtwo{M_j}{I}{t}=\mfalse$ for all $j\in\{k+1,\ldots, m\}$. By the same reasoning performed at the point above, it follows that $\semtwo{M}{J}{t}=\mtrue$. By Lemma \ref{lemma:twomonotonicity}, we get that $\semtwo{M}{I}{t}=\mtrue$, as desired.
 \end{enumerate}
 Thus, we can conclude that $(J,I)$ is an HT-model. Since $I$ is a stable HT-model, it follows that $J=I$, i.e.\ $I$ is the least fixpoint of $ \appr{\db}{\Pi}^1(\cdot,I)$.

Assume now that $(I,I)$ is a three-valued stable model of $\db$ and $\Pi$. By Definition \ref{def:models}, it holds that $I$ is a fixpoint of both $\appr{\db}{\Pi}^1(\cdot,I)$ and $\appr{\db}{\Pi}^2(I,\cdot)$, i.e.\ $\appr{\db}{\Pi}(I,I)=(I,I)$. By Proposition  \ref{prop:HTiffthreevalued}, $(I,I)$ is an HT-model of $\db$ and $\Pi$. It remains to show that the equality \eqref{eq:glbWalegastable} in Definition \ref{def:stablewalega} holds. Let $J\subseteq I$ such that $(J,I)$ is an HT-model of $\db$ and $\Pi$. We show that indeed $J$ must be equal to $I$. 
It suffices to prove that $I\subseteq J$. 

 By Proposition 3.3 of  \cite{DeneckerMT04}, $\appr{\db}{\Pi}^1(\cdot,I)$ is a monotone operator on the complete lattice $\langle \{H\mid H\subseteq I\}, \subseteq\rangle$. Hence, by the Tarski–Knaster Theorem \cite{Tarski55}, the least fixpoint of $\appr{\db}{\Pi}^1(\cdot,I)$ coincide with its least pre-fixpoint. Since $\lfp \appr{\db}{\Pi}^1(\cdot,I)=I$ by hypothesis, it follows that $I$ is also the least pre-fixpoint of $\appr{\db}{\Pi}^1(\cdot,I)$. By Lemma \ref{lemma:HTprefixpt}, $J$ is a prefixpoint of $\appr{\db}{\Pi}^1(\cdot, I)$. Thus, it must be the case that $I\subseteq J$, as desired.
\end{proof}

\section{Conclusions and Future Work}

Using the solid framework of Approximation Fixpoint Theory, we have defined several semantics of interest for DatalogMTL$\neg$: the Kripke-Kleene, well-founded, stable, and supported model semantics. This was possible by defining a \emph{consistent approximator}, i.e.\ an operator on the space of partial, three-valued interpretations. As we have shown, such an operator does not require any farfetched choice; on the contrary, it is the most natural adaptation of Fitting's three-valued immediate consequence operator to the setting of DatalogMTL$\neg$, showcasing, yet again, the simplicity, versatility, and power of the AFT formalism. In the final section, we have also proved that the stable models obtained via AFT are equivalent to the ones defined in \cite{WalegaCGK24} through interpretations stemming from the logic of HT.

Since pairs of interpretations are used both in AFT and for the logic of HT, although in a different way, it is often the case that when the logic of HT is employed to define some \emph{stable models}, AFT can provide an easier route to achieve the same results, while also giving access to other types of models for free. This might be true for the \emph{metric stable models} defined in \cite{BeckerCDSS24}, used to obtain the non-monotonic Metric Equilibrium Logic (MEL) starting from the Metric logic of HT, a temporal extension of the logic of HT. If this is indeed the case, AFT could provide a solid and reliable ground for a comparison between MEL developed in \cite{BeckerCDSS24}, and  the extension of Datalog MTL with negation-as-failure of \cite{WalegaCGK24}, discussed in this paper.

It is also worth mentioning that AFT grants access to a wide body of theoretical results that hold for its applications. A line of work focused on stratification results \cite{VennekensGD06,BogaertsC21}. In particular, a notion of \emph{stratifiable operator} was defined, allowing for the computation of the Kripke-Kleene, supported, well-founded, and stable models stratum by stratum. This seems to be equivalent to what was done for stratifiable DatalogMTL$\neg$, i.e.\ for DatalogMTL$\neg$ programs with no cyclic dependencies via negation \cite{CucalaWGK21}.

\nocite{*}
\bibliographystyle{eptcs}
\bibliography{mybib}
\end{document}